\documentclass[graybox]{svmult}

\smartqed
\usepackage{mathptmx}       
\usepackage{helvet}         
\usepackage{courier}        
\usepackage{type1cm}        
\usepackage{graphicx}       

\usepackage{array,colortbl}
\usepackage{amsmath,amsfonts,amssymb,bm} 
\DeclareFontFamily{U}{mathx}{\hyphenchar\font45}
\DeclareFontShape{U}{mathx}{m}{n}{
      <5> <6> <7> <8> <9> <10>
      <10.95> <12> <14.4> <17.28> <20.74> <24.88>
      mathx10
      }{}
\DeclareSymbolFont{mathx}{U}{mathx}{m}{n}
\DeclareFontSubstitution{U}{mathx}{m}{n}
\DeclareMathAccent{\widecheck}      {0}{mathx}{"71}

\renewcommand{\email}[1]{\emailname: #1} 

\usepackage{mathptmx}       
\usepackage{helvet}         
\usepackage{courier}        

\usepackage{makeidx}         
\usepackage{graphicx}        
\usepackage[bottom]{footmisc}

\usepackage{latexsym}
\usepackage{amsmath}
\usepackage{amsfonts}
\usepackage{amssymb}
\usepackage{bm}

\usepackage{url}
\usepackage{algorithm}
\usepackage{algorithmic}
\usepackage[misc,geometry]{ifsym}

\renewenvironment{proof}{\noindent{\itshape Proof.}}{\smartqed\qed}

\spdefaulttheorem{assumption}{Assumption}{\upshape \bfseries}{\itshape}
\spdefaulttheorem{algo}{Algorithm}{\upshape \bfseries}{\itshape}

\newcommand{\bbE}{{\mathbb{E}}}

\newcommand{\bbP}{{\mathbb{P}}}

\newcommand{\bbR}{{\mathbb{R}}}

\DeclareSymbolFont{bbold}{U}{bbold}{m}{n}
\DeclareSymbolFontAlphabet{\mathbbold}{bbold}
\newcommand{\ind}{{\mathbbold{1}}}

\newcommand{\ceil}[1]{\left\lceil #1 \right\rceil}    

\DeclareMathOperator{\var}{Var}

\DeclareSymbolFont{bbold}{U}{bbold}{m}{n}
\DeclareSymbolFontAlphabet{\mathbbold}{bbold}

\usepackage{microtype} 

\usepackage[colorlinks=true,linkcolor=black,citecolor=black,urlcolor=black]{hyperref}
\usepackage{bookmark}
\makeatletter 
  \providecommand*{\toclevel@author}{999}
  \providecommand*{\toclevel@title}{0}
\makeatother

\usepackage{lmodern}
\usepackage{amssymb,amsmath}
\usepackage{ifxetex,ifluatex}
\usepackage{fixltx2e} 
\ifnum 0\ifxetex 1\fi\ifluatex 1\fi=0 
  \usepackage[T1]{fontenc}
  \usepackage[utf8]{inputenc}
\else 
  \ifxetex
    \usepackage{mathspec}
  \else
    \usepackage{fontspec}
  \fi
  \defaultfontfeatures{Ligatures=TeX,Scale=MatchLowercase}
\fi
\IfFileExists{upquote.sty}{\usepackage{upquote}}{}
\IfFileExists{microtype.sty}{%
\usepackage{microtype}
\UseMicrotypeSet[protrusion]{basicmath} 
}{}
\usepackage[margin=1in]{geometry}
\usepackage{hyperref}
\hypersetup{unicode=true,
            pdftitle={Robust estimation of the mean with bounded relative standard deviation},
            pdfauthor={Mark Huber},
            pdfborder={0 0 0},
            breaklinks=true}
\usepackage{color}
\usepackage{fancyvrb}

\DefineVerbatimEnvironment{Highlighting}{Verbatim}{commandchars=\\\{\}}
\usepackage{framed}
\definecolor{shadecolor}{RGB}{248,248,248}
\newenvironment{Shaded}{\begin{snugshade}}{\end{snugshade}}
\newcommand{\KeywordTok}[1]{\textcolor[rgb]{0.13,0.29,0.53}{\textbf{{#1}}}}
\newcommand{\DataTypeTok}[1]{\textcolor[rgb]{0.13,0.29,0.53}{{#1}}}
\newcommand{\DecValTok}[1]{\textcolor[rgb]{0.00,0.00,0.81}{{#1}}}

\newcommand{\FloatTok}[1]{\textcolor[rgb]{0.00,0.00,0.81}{{#1}}}

\newcommand{\CharTok}[1]{\textcolor[rgb]{0.31,0.60,0.02}{{#1}}}

\newcommand{\StringTok}[1]{\textcolor[rgb]{0.31,0.60,0.02}{{#1}}}

\newcommand{\CommentTok}[1]{\textcolor[rgb]{0.56,0.35,0.01}{\textit{{#1}}}}

\newcommand{\NormalTok}[1]{{#1}}
\usepackage{longtable,booktabs}
\usepackage{graphicx,grffile}
\makeatletter
\def\maxwidth{\ifdim\Gin@nat@width>\linewidth\linewidth\else\Gin@nat@width\fi}
\def\maxheight{\ifdim\Gin@nat@height>\textheight\textheight\else\Gin@nat@height\fi}
\makeatother

\begin{document}

\title*{Robust estimation of the mean with bounded relative standard deviation}
\author{Mark Huber}
\institute{
Mark Huber 
\at Claremont McKenna College, 850 Columbia Avenue, Claremont, CA, USA
\email{mhuber@cmc.edu}
}

\maketitle


\abstract{
Many randomized approximation algorithms operate by giving a procedure for simulating a random variable $X$ which has mean $\mu$ equal to the target answer, and a relative standard deviation bounded above by a known constant $c$.  Examples of this type of algorithm includes methods for approximating the number of satisfying assignments to 2-SAT or DNF, the volume of a convex body, and the partition function of a Gibbs distribution.  Because the answer is usually exponentially large in the problem input size, it is typical to require an estimate $\hat \mu$ satisfy $\bbP(|\hat \mu/\mu - 1| > \epsilon) \leq \delta$, where $\epsilon$ and $\delta$ are user specified nonnegative parameters.  The current best algorithm uses $\ceil{2c^2\epsilon^{-2}(1+\epsilon)^2 \ln(2/\delta)}$ samples to achieve such an estimate.  By modifying the algorithm in order to balance the tails, it is possible to improve this result to $\ceil{2(c^2\epsilon^{-2} + 1)/(1-\epsilon^2)\ln(2/\delta)}$ samples.  Aside from the theoretical improvement, we also consider how to best implement this algorithm in practice.  Numerical experiments show the behavior of the estimator on distributions where the relative standard deviation is unknown or infinite.
}


\section{Introduction}

Suppose we are interested in approximating a target value $\mu$.  Then many randomized approximation algorithms work by constructing a random variable $X$ such that $\bbE[X] = \mu$ and $\var(X)/\mu^2 \leq c^2$ for a known constant $c$.  The randomized algorithm then simulates $X_1,X_2,\ldots,X_n$ as independent identically distributed (iid) draws from $X$.  Finally, the values are input into a function to give an estimate $\hat \mu$ for $\mu$.

Examples of this type of algorithm include when $\mu$ is the number of solutions to a logic formula in Disjunctive Normal Form (DNF)~\cite{karpl1983}, the volume of a convex body~\cite{dyerfk1991}, and the partition function of a Gibbs distribution~\cite{huber2015a}.  For all of these problems, the random variable $X$ used is nonnegative with probability 1, and so we shall only consider this case for the rest of the paper.

For these types of problems, generating the samples from a high dimensional distribution is the most computationally intensive part of the algorithm.  Hence we will measure the running time of the algorithm by the number of samples from $X$ that must be generated.

The answer $\mu$ for these problems typically grows exponentially quickly in the size of the problem input size.  Therefore, it is usual to desire an approximation $\hat \mu$ that is accurate when measured by relative error.  We use $\epsilon > 0$ as our bound on the relative error, and $\delta > 0$ as the bound on the probability that the relative error restriction is violated.

\begin{definition}
Say $\hat \mu$ is an \emph{$(\epsilon,\delta)$-randomized approximation scheme} ($(\epsilon,\delta)$-ras) if 
\[
\bbP\left(\left|\frac{\hat \mu}{\mu} - 1\right| > \epsilon\right) \leq \delta.
\]
\end{definition}

This is equivalent to considering a loss function $L$ that is $L(\hat \mu,\mu) = \ind(|(\hat \mu/\mu)-1| > \epsilon)$, and requiring that the expected loss be no more than $\delta$.

Note that this question is slightly different than the problem most statistical estimates are designed to handle.  For instance, the classic work of~\cite{huber1964} is trying to minimize the asymptotic variance of the estimator, not determine how often the relative error is at most $\epsilon$.

The first approach to this problem is often the standard sample average
\[
S_n = \frac{X_1 + \cdots + X_n}{n}.
\]
For this estimator, only knowing that $\var(X) \leq c^2 \mu^2$, the best we can say about the probability that the relative error is large comes from Chebyshev's inequality~\cite{tchebichef1867}.
\[
\bbP(|S_n - \mu| > \epsilon \mu) \leq \frac{\var(X)}{n \epsilon^2 \mu^2} \leq c^2 \epsilon^{-2} n^{-1}.
\]

In particular, setting $n = c^2 \epsilon^{-2} \delta^{-1}$ gives an $(\epsilon,\delta)$-ras.  There are simple examples where Chebyshev's inequality is tight.  

While this bound works for all distributions, for most distributions in practice the probability of error will go down exponentially (and not just polynomially) in $n$.  We desire an estimate that matches this speed of convergence.

For example, suppose that $X$ is a normal random variable with mean $\mu$ and variance $c^2 \mu^2$ (write $X \sim \textsf{N}(\mu,c^2 \mu^2)$.  Then the sample average is normally distributed as well.  To be precise, $S_n \sim \textsf{N}(\mu,c^2\mu^2/n)$, and it is straightforward to show that 
\[
\bbP(|S_n - \mu| > \epsilon \mu) = \bbP(|Z| \geq \epsilon \sqrt{n} / c), 
\]
where $Z$ is a standard normal random variable.
This gives 
\begin{equation}
    \label{EQN:bound}
n = 2 c^2 \epsilon^{-2} [\ln(2/\delta) - o(\delta)].
\end{equation}
samples being necessary and sufficient to achieve an $(\epsilon,\delta)$-ras. 

We did this calculation for normally distributed samples, but in fact this gives a lower bound on the number of samples needed.  For any $X_1,\ldots,X_n$ with mean $\mu > 0$ and variance $c^2 > 0$, and an estimate $\hat \theta = \hat \theta(X_1,\ldots,X_n)$,
\[
\bbP(|\hat \theta - \mu| > \epsilon \mu) \geq (1/2) \bbP(|Z| \geq \epsilon \sqrt{n} / c).
\]
See, for instance, Proposition 6.1 of~\cite{catoni2012}.  Hence the minimum number of samples required for all instances is 
\[
n \geq 2 c^2 \epsilon^{-2}[\ln(1/\delta) - o(\delta)].
\]

The method presented in~\cite{hubertoappearb} comes close to this lower bound, requiring
\[
n = \lceil 2 c^2 \epsilon^{-2}(1 + \epsilon)^2 \ln(2/\delta) \rceil 
\]
samples.  So it is larger than optimal by a factor of $(1+\epsilon)^2$.  

Our main result is to modify the estimator slightly.  This has two beneficial effects.
\begin{enumerate}
  \item  The nuisance factor is reduced from first order in $\epsilon$ to second order.  To be precise, the new nuisance factor is $(1 + \epsilon^2/c^2) / (1 - \epsilon^2)$.
  \item It is possible to solve exactly for the value of the $M$-estimator (using square and cube roots) rather than through numerical appoximation if so desired.
\end{enumerate}

\begin{theorem}
\label{THM:main}
Suppose that $\bbP(X \geq 0) = 1$, $\bbE[X] = \mu$, and the standard deviation of $X$ is at most $c \mu$, where $c > 0$.  Then there exists an $(\epsilon,\delta)$-ras $\hat\mu_1$ where at most 
\[
\ceil{2 (c^2 \epsilon^{-2} + 1)(1-\epsilon^2)^{-1}\ln(2/\delta)}
\]
draws from $X$ are used.
\end{theorem}

Ignoring the ceiling function, the new method uses a number of samples bounded by the old method times a factor of
\[
\frac{1 + \epsilon^2/c^2}{(1 + \epsilon)^2 (1 - \epsilon^2)}.
\]
For instance, when $\epsilon = 0.1$ and $c = 2$, this factor is $0.8556\ldots$, and so we obtain an improvement of over 14\% in the running time.  At first, this might seem slight, but remember that the best improvement we can hope to make based on normal random variables is $1/(1+\epsilon)^2 = 0.8264$, or a bit less than 18\%.  Therefore, this does not quite obtain the maximum improvement of $1/(1 + 2\epsilon + \epsilon^2)$, but does obtain a factor of $(1 + O(\epsilon^2))/(1 + 2\epsilon + \epsilon^2)$.

The remainder of this paper is organized as follows.  In Section~\ref{SEC:catoni}, we describe what $M$-estimators are, and show how to down weight samples that are far away from the mean.  Section~\ref{SEC:computation} shows how to find the estimator both approximately and exactly.  In Section~\ref{SEC:experiments}  we consider using these estimators on some small examples and see how they behave numerically.

\section{\texorpdfstring{$\Psi$}{Psi}-estimators}
\label{SEC:catoni}

The median is a robust centrality estimate, but it is easier to build a random variable whose mean is the target value.  P.J. Huber first showed in~\cite{huber1964} how to build an estimate that has the robust nature of the median estimate while still converging to the mean.

Consider a set of real numbers $x_1,\ldots,x_n$.
The sample average of the $\{x_i\}$ is the point $m$ where the sum of the distances to points to the right of $m$ equals the sum of the distances to points to the left of $m$.  

This idea can be generalized as follows.  First, begin with a function $\psi:\bbR \times \bbR^n \rightarrow \bbR$.  Second, using $\psi$, form the function
\[
\Psi(m) = \sum_{i=1}^n \psi(x_i, m).
\]
For convenience, we suppress the dependence of $\Psi$ on $(x_1,\ldots,x_n)$ in the notation.

Consider the set of zeros of $\Psi$.  These zeros form the set of \emph{$\Psi$-type $M$-estimators} for the center of the points $(x_1,\ldots,x_n)$.

For example, suppose our $\psi$ is $f:\bbR^2 \rightarrow \bbR$ defined as 
\[
f(x,m) = \left\{ \begin{array}{ll} (x/m) - 1 \hspace*{2em} & m \neq 0 \\
0 & x = m = 0.\end{array} \right.
\]
Then consider 
\[
\sum_{i = 1}^m f(x_i, m) = 0.
\]

For $x$ and $m$ with the same units, the right hand side is unitless.  If the $x_i$ are nonnegative $\{x_i\}$ and at least one $x_i$ is positive, 
\[
\sum_{i=1}^n f(x_i,m) = 0 \Leftrightarrow m = \frac{\sum_{i=1}^n x_i}{n}.
\]
That is, the unique $M$-estimator using $f$ as our $\psi$ function is the sample average.

Now consider 
\[
g(x,m) = \ind\left(f(x,m) > 0\right) - \ind\left(f(x,m) < 0\right).
\]
Then we wish to find $m$ such that
\[
\sum_{i=1}^n g(x_i,m) = 0.
\]

Summing $g(x_i, m)$ adds 1 when $x_i$ greater than $m$, and subtracts 1 when $x_i$ is $m$.  Then the sum of the $g(x_i,m)$ is zero exactly when there are an equal number of $x_i$ that are above and below $m$.  When the number of distinct $x_i$ values is odd, then $m$ has a unique solution equal to the sample median.

The sample median has the advantage of being robust to large changes in the $x_i$ values, but this will converge to the median of the distribution (under mild conditions) rather than the mean.  The sample average actually converges to the mean when applied to iid draws from $X$, and this is the value we care about finding.  However, the sample average can be badly thrown off by a single large outlier.  Our goal is to create a $\Psi$ function that combines the good qualities of each while avoiding the bad qualities.

Both $f$ and $g$ can be expressed using weighted differences.  Let $d_f:\bbR \rightarrow \bbR$ and $d_g:\bbR \rightarrow \bbR$ be defined as
\begin{align*}
    d_f(u) &= u \\
    d_g(u) &= \ind(u > 0) - \ind(u < 0).
\end{align*}
Let $u(x_i, m) = x_i / m - 1$ for $m \neq 0$, and $0$ for $x_i = m = 0$.  Then $f(x_i, m) = d_f(u)$ and $g(x_i, m) = d_g(u)$.

Catoni~\cite{catoni2012} improved upon this $\Psi$-estimator by using a function that approximated $d_f$ for $|u| \leq 1$, and approximated $d_u$ for $|u| > 1$, and could be analyzed using the Chernoff bound approach~\cite{chernoff1952}.   Catoni and Guillini then created an easier to use version of the $\Psi$-estimator in~\cite{catonig2017}.

The following is a modification of the Catoni and Guillini $\Psi$-estimator.  Unlike~\cite{catonig2017}, the weighted difference here does not have square roots in the constants, which makes them slightly easier to work with from a computational standpoint. \begin{align*}
    d_h(u) &= \left(\frac{5}{6}\right)\ind(u > 1) + \left(u - \frac{u^3}{6}\right)\ind(u \in [-1,1]) - \left(\frac{5}{6}\right) \ind(u < -1).
\end{align*}
Then define
\[
h(x_i, m) = d_h(u(x_i, m )).
\]

\begin{figure}[h!]
\begin{center}
\includegraphics[clip=true,trim=1.8in 5.4in 1.8in 2in,scale=0.7]{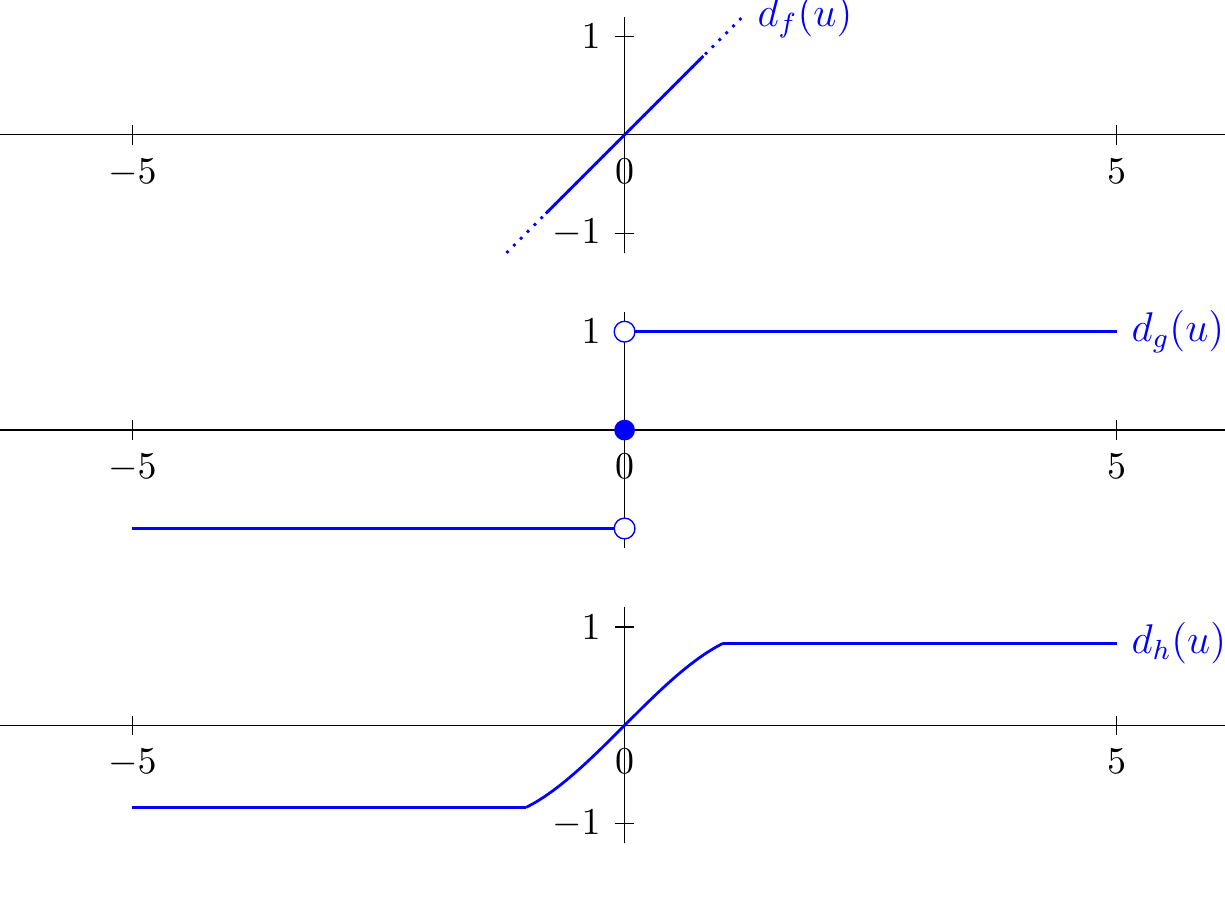}
\caption{The functions $d_f$, $d_g$, and $d_h$.}
\label{FIG:d3}
\end{center}
\end{figure}

The function $d_h$ behaves like the mean weight for values near 0, but like the median weights for values far away from 0.  See Figure~\ref{FIG:d3}.
In order to link this function to the first and second moments of the random variable, the following links $d_h$ to weighted differences introduced in~\cite{catoni2012}.  See Figure~\ref{FIG:bounds}.

\begin{lemma}
    Let 
    \begin{align*}
        d_L(u) &= -\ln(1-u+u^2/2) \\
        d_U(u) &= \ln(1+u+u^2/2).
    \end{align*}
    Then for all $u \in \bbR$,
    \[
    d_L(u) \leq d_h(u) \leq d_U(u)
    \]
\end{lemma}

\begin{proof}
    Note $d'_U(u) = (1+u)/(1+u+u^2/2)$.  This derivative is positive over $[1,\infty)$ so $d_U$ is increasing in this region.  Since $d_U(1) = \ln(2.5) > 5/6 = d_h(1)$, $d_U(u) \geq d_h(u)$ over $[1,\infty)$.  
    
    Similarly, $d'_U(u) < 0$ for $u \in (-\infty,-1]$, and so it is decreasing in this region, and $d_U(-1) = \ln(1/2) \geq -5/6 = d_h(u)$ so $d_U(u) \geq d_h(u)$ in $(-\infty,1]$.  
    Finally, inside $[-1,1]$, $d_h(u) = u-u^3/6$.  So a minimum of $d_U(u) - d_h(u)$ occurs either at $-1$, $1$, or a critical point where $d'_U(u) - (1-u^2/2) = 0$.  The unique critical point in $[-1,1]$ is at $u = 0$, and $d_U(0) = d_h(0)$.  This means $d_U(u) - d_h(u) \geq 0$ for all $u \in [-1,1]$.
    
    The lower bound follows from $d_L(u) = -d_U(-u) \leq -d_h(-u) = d_h(u)$.
\end{proof}

\begin{figure}[h!]
\begin{center}
\includegraphics[clip=true,trim=1.8in 2.8in 0.7in 5.7in,scale=0.7]{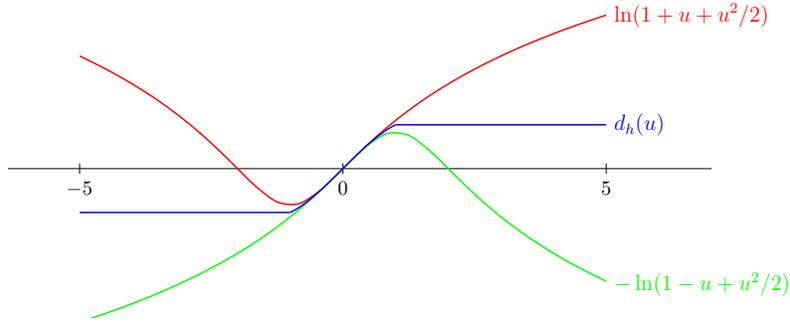}
\caption{Lower and upper bounds for $d_h$.}
\label{FIG:bounds}
\end{center}
\end{figure}

It helps to introduce a scale factor $\lambda$ that allows us to extend the effective range where $d(u) \approx u$.  Let
\[
h_\lambda(x_i,m) = \lambda^{-1} d_h(\lambda(x_i/m - 1)).
\]
Then $\lambda > 0$ is a parameter that can be chosen by the user ahead of running the algorithm based upon $\epsilon$ and $c$.

The following was shown as Lemma 13 of~\cite{hubertoappearb}.
\begin{lemma}
For given $\epsilon, \delta > 0$, let $\epsilon' = \epsilon/(1+\epsilon)$ and $\lambda = \epsilon'/c^2$.  Set 
\[
n = \ceil{2c^2\epsilon^{-2}\ln(2/\delta)(1+\epsilon)^2}.
\]
For $X_1,\ldots,X_n$ iid $X$, form the function
\[
\Psi_\lambda(m) = \frac{1}{n} \sum_{i=1}^n \lambda^{-1} d_h(\lambda u(X_i, m)).
\]
Let $\hat \mu$ be any value of $m$ such that $|\Psi_\lambda(m)| \leq (\epsilon')^2/2.$  Then $\hat \mu$ is an $(\epsilon,\delta)$-ras for $\mu$. 
\end{lemma}
To improve upon this result and obtain Theorem~\ref{THM:main}, we must be more careful about our choice of $\lambda$.

To meet the requirement of an $(\epsilon,\delta)$-ras, we must show that $\Psi_{\lambda}(m)$ has all its zeros in the interval $((1-\epsilon)\mu,(1+\epsilon)\mu)$ with probability at least $1 - \delta$.  Given that $\Psi_\lambda$ is continuous and decreasing, it suffices to show that $\Psi_{\lambda}((1-\epsilon)\mu) > 0$ and $\Psi_{\lambda}((1+\epsilon)\mu) < 0$.

\begin{lemma}
\label{LEM:chernoff}
Let $X_1,\ldots,X_n$ be iid $X$.  Then for $\lambda > 0$,
\begin{align*}
\bbP(\Psi_{\lambda}((1+\epsilon)\mu) \geq 0) &\leq \left[1 + \bbE[u_{\epsilon}] + \frac{1}{2}\bbE[u^2_{\epsilon}]\right]^n, \\
\bbP(\Psi_{\lambda}((1-\epsilon)\mu) \leq 0)  &\leq \left[1 - \bbE[u_{-\epsilon}] + \frac{1}{2}\bbE[u^2_{-\epsilon}]\right]^n.
\end{align*}
where
\[
u_{\epsilon} = \frac{X}{(1+\epsilon)\mu} - 1.
\]
\end{lemma}

\begin{proof}
    Note that 
    \begin{align*}
        \bbP(\Psi_{\lambda}((1+\epsilon)\mu) \geq 0) &= \bbP(\exp(\lambda \Psi_{\lambda}((1+\epsilon)\mu) \geq 1) \\
        &\leq \bbE[\exp(\lambda \Psi_{\lambda}((1+\epsilon)\mu))]
    \end{align*}
    by Markov's inequality.  Note
    \[
    \exp(\lambda \Psi_{\lambda}(m)) = \prod_{i=1}^n \exp(\lambda d_h(X_i/m-1)).
    \]
    Each term in the product is independent, therefore the mean of the product is the product of the means (which are identical.)
    
    Let $u = \lambda(X/m - 1)$.  Then
    \[
    \bbP(\Psi_{\lambda}((1+\epsilon)\mu) \geq 0) \leq [\bbE(\exp(\lambda d_h(u)))]^n
    \]
    
    From the previous lemma,
    \begin{align*}
        \exp(d_h(u)) &\leq \exp(\ln(1+u+u^2/2)) \\
        &= 1 + u + u^2/2,
    \end{align*}
    Hence
    \[
    \bbE(\exp(d_h(u))) = 1 + \bbE(u) + \bbE(u^2/2).
    \]
    Putting $m = (1+\epsilon)\mu$ into this expression then gives the first inequality.  
    
    For the second inequality, the steps are nearly identical, but we begin by multiplying by $-\lambda$.  This completes the proof.
\end{proof}

In particular, if we choose $n$ so that $\bbP(\Psi_\lambda((1+\epsilon)\mu) > 0) \leq \delta/2$, and $\bbP(\Psi_\lambda((1-\epsilon)\mu) < 0) \leq \delta/2$, then by the union bound the probability that $\Psi$ has a root in $[(1-\epsilon)\mu,(1+\epsilon)\mu]$ is at least $\delta/2 + \delta/2=\delta$.

As is well known, for $g > 0$,
\[
(1-g)^n \leq \exp(-gn),
\]
and so if $n \geq (1/g)\ln(2/\delta)$, 
\[
(1-g)^n \leq \delta/2.
\]
We refer to $g$ as the \emph{gap}.  Since the number of samples needed is inversely proportional to the gap, we wish the gap to be as large as possible.
We can lower bound the gap given by the previous lemma in terms of $\lambda$ and $m$.
\begin{lemma}
\label{LEM:bound}
  For $u = \lambda(X/m-1)$, and
  $a(m) = 1 - (m/\mu)$, 
  \begin{equation}
      \label{EQN:plusside}
  1 + \bbE[u] + \bbE[u^2/2] \leq 
  1 + \frac{\lambda \mu}{m}a(m) + \frac{1}{2}\left(\frac{\lambda \mu}{m}\right)^2\left[c^2 + a(m)^2\right].
  \end{equation}
    Similarly,
    \begin{equation}
        \label{EQN:negside}
    1 + \bbE[u] + \bbE[u^2/2] \leq 
    1 - \frac{\lambda \mu}{m}a(m) + \frac{1}{2}\left(\frac{\lambda \mu}{m}\right)^2\left[c^2 + a(m)^2\right].
    \end{equation}    
\end{lemma}

\begin{proof}
  From linearity of expectations, 
  \[
  \bbE[u] = \lambda\left(\frac{\mu}{m}-1\right) = \frac{\lambda \mu}{m} a(m).
  \]
  The  second moment is the sum of the variance plus the square of the first moment.  Using $\var(c_1 X + c_2) = c_1^2 \var(X)$, we get 
    \begin{align*}
        \bbE[u^2] &= \var(u) + \bbE[u]^2 \\
        &= \left(\frac{\lambda}{m}\right)^2\var(X)+ \left[\frac{\lambda \mu}{m} a(m) \right]^2.
    \end{align*}
    Since $\var(X) \leq c^2\mu^2$, 
    \begin{align*}
        \bbE[u^2] &\leq \left(\frac{\lambda \mu}{m}\right)^2 [c^2 + a(m)^2],
    \end{align*}
    giving the first result.  The proof of the second statement
    is similar.
\end{proof}

\begin{lemma}
\label{LEM:upperbound}
Let $p_\ell = \bbP(\Psi_{\lambda}((1-\epsilon)\mu) \leq 0)$ and $p_r = \bbP(\Psi_{\lambda}((1+\epsilon)\mu) \geq 0)$.  Then
  \begin{align*}
p_\ell &\leq \left[1 - \frac{\lambda \epsilon}{1-\epsilon} + \frac{1}{2}\left(\frac{\lambda}{1-\epsilon}\right)^2 (c^2 + \epsilon^2)\right]^n, \\
p_r &\leq \left[1 - \frac{\lambda \epsilon}{1+\epsilon} + \frac{1}{2}\left(\frac{\lambda}{1+\epsilon}\right)^2 (c^2 + \epsilon^2)\right]^n.
  \end{align*}
\end{lemma}

\begin{proof}  Note 
\[
a((1-\epsilon)\mu) = 1-(1-\epsilon)\mu/\mu = \epsilon.
\]
      Combine Lemmas~\ref{LEM:chernoff},~\ref{LEM:bound}, and~\ref{LEM:min} with $m$ equal to $(1-\epsilon)\mu$ to get the first inequality.  Then set $m = (1+\epsilon)\mu$ and use $a((1+\epsilon)\mu) = -\epsilon$ to get the second.
\end{proof}

Our goal is to use as few samples as possible, which means simultaneously minimizing the quantities inside the brackets in Lemma~\ref{LEM:upperbound}.  
Both of the upper bounds are upward facing parabolas, and they have a unique minimum value.  
\begin{lemma}
\label{LEM:min}
  The minimum value of $f(\lambda) = 1 + a_1 \lambda + (1/2)a_2 \lambda^2$ is 
  \[
  1 - \frac{a_1^2}{2a_2},
  \]
  at $\lambda^* = -a_1/a_2.$
\end{lemma}

\begin{proof} Complete the square. \end{proof}

However, because the coefficients in the quadratic upper bounds are different, it is not possible to simultaneously minimize these bounds with the same choice of $\lambda$.  

What can be said is that these bounds are quadratic with positive coefficient on $\lambda^2$, and so the best we can do is to choose a $\lambda$ such that the upper bounds are equal to one another.  This gives us our choice of $\lambda$.

\begin{lemma}
    Let  
    \[
    \lambda = \frac{\epsilon}{c^2+\epsilon^2}(1-\epsilon^2).
    \]
    Then 
    \begin{align*}
        \max\{p_{\ell},p_r\} \leq \left[1 - \frac{1}{2} \cdot \frac{\epsilon^2}{c^2+\epsilon^2}(1-\epsilon^2) \right]^n. 
    \end{align*}
\end{lemma}

\begin{proof}
   Follows directly from Lemma~\ref{LEM:upperbound}.
\end{proof}

This makes the inverse gap
\[
\epsilon^{-2}(c^2 + \epsilon^2)(1-\epsilon^2)^{-1}
= (c^2\epsilon^{-2} + 1)(1 - \epsilon^{-2})^{-1}
\]
and immediately gives Theorem~\ref{THM:main}.

\section{Computation}
\label{SEC:computation}

Set $\lambda = \epsilon(1-\epsilon^2)(c^2+\epsilon^2)^{-1}$.
Consider how to locate any root of $\Psi_\lambda$ for a given set of $X_1,\ldots,X_n$.  As before, we assume that the $X_i$ are nonnegative and not all identically zero.  The function $\Psi_\lambda(m)$ is continuous and decreasing, although not necessarily strictly decreasing.  Therefore, it might have a set of zeros that form a closed interval.

\subsection{An approximate method}

Suppose that the points $X_i$ are sorted into their order statistics,
\[
X_{(1)} \leq X_{(2)} \leq \cdots X_{(n)}.
\]
Then $\Psi_\lambda(X_{(1)}) > 0$ and $\Psi_\lambda(X_{(n)}) < 0$, so there exists some $i$ such that $\Psi_{\lambda}(X_{(i)}) \geq 0$ and $\Psi_{\lambda}(X_{i+1}) \leq 0$.  Since any particular value of $\Psi_\lambda$ requires $O(n)$ time to compute, this index $i$ can be found using binary search in $O(n\ln(n))$ time.

At this point we can switch from a discrete binary search over $\{1,\ldots,n\}$ to a continuous binary search over the interval $[X_{i},X_{i+1}]$.  This allows us to quickly find a root to any desired degree of accuracy.

This is the method that would most likely be used in practice.

\subsection{An exact method}

Although the approximation procedure is what would be used in practice because of its speed, there does exist a polynomial time exact method for this problem.
For $i < j$, and a value $m$, suppose that the points of $\{X_i\}$ that fall into the interval $[m - \lambda,m + \lambda]$ are exactly $\{X_{(i)},X_{(i+1)},\ldots,X_{(j)}\}.$  Then say that $m \in m(i,j)$.

That is, define the set $m(i,j)$ for $i < j$ as follows. 
\[
m(i,j) = \{m:X_{(k)} \in [m-\lambda,m+\lambda] \Leftrightarrow k \in \{i,i+1,\ldots,j\} \}.
\]
See Figure~\ref{FIG:interval}.

\begin{figure}[h!]
\begin{center}
\includegraphics[clip=true,trim=1.8in 8.7in 1.8in 1.7in,scale=0.7,page=2]{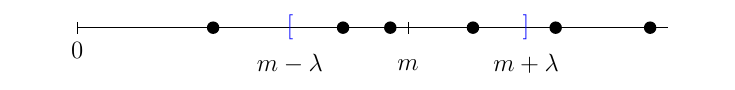}
\caption{Since the interval $[m-\lambda,m+\lambda]$ includes $X_{(2)}$, $X_{(3)}$, and $X_{(4)}$, $m \in m(2,4)$.}
\label{FIG:interval}
\end{center}
\end{figure}

Note that 
\[
(\forall m \in [X_{(1)},X_{(n)}])(\exists i < j)(m \in m(i,j)).
\]
There are at most $n$ choose 2 such $i < j$ where $m(i,j)$ is nonempty.  In fact, since as we slide $m$ from $X_{(1)}$ to $X_{(n)}$, each point can enter or leave the interval $[m-\lambda,m+\lambda]$ exactly once.  Therefore there are at most $2n$ pairs $(i,j)$ where $m(i,j)$ is nonempty.

That means that to find a root of $\Psi_\lambda$ we need merely check if there is a root $m \in m(i,j)$ for all $i < j$ such that $m(i,j) \neq \emptyset$.

For each $m \in m(i,j)$, the contribution of $X_{(1)},\ldots,X_{(i-1)}$ to $\Psi_\lambda(m)$ is $-(i-1)$, and the contribution of $X_{(j+1)},\ldots,X_{(n)}$ is $n - j$.  Hence $m \in m(i,j)$ is a zero of $\Psi_\lambda$ if an only if $m = 1/r$ and 
\[
n - j - (i - 1) + \sum_{k=i}^j (rX_{(k)} - 1) - (r X_{(k)} - 1)^3/6 = 0.
\]
This last equation is a cubic equation in $r$, and so the value of $r$ that satisfies it (assuming such exists) can be determined exactly using the cubic formula.

\section{Numerical Experiments}
\label{SEC:experiments}

The $M$-estimator presented here can be thought of as a principled interpolation between the sample mean and the sample median.  Because for $u$ small $d(u) \approx u$, as $\lambda \rightarrow 0$, the estimator converges to the sample mean.

At the other extreme, as $\lambda \rightarrow \infty$, all of the $d(\lambda(X_i/m-1)$ will evaluate to either 0, 1, or -1.  Hence the estimator converges to the sample median.

When $\lambda = \epsilon(1-\epsilon^2)/(c^2+\epsilon^2)$, we can precisely bound the chance of the relative error being greater than $\epsilon$.  However, the estimator can be used for any value of $\lambda$.

For instance, Table~\ref{TBL:1} records the result of using the estimator for 100 exponential random variables with mean 2 and median $2\ln(2) = 1.386\ldots$.

\begin{table}
\caption{Behavior of the sample average, sample mean, and $M$-estimator for 100 exponential random variables with mean 2. here.  Repeated five times to show variation.}\label{TBL:1}
%
%
\begin{tabular}{p{1.6cm}p{1.6cm}p{1.6cm}p{1.6cm}p{1.6cm}}
\hline\noalign{\smallskip}
Mean & Median & $\lambda = 0.1$ & $\lambda = 1$ & $\lambda = 5$ 
 \\
\noalign{\smallskip}\svhline\noalign{\smallskip}
2.34 & 1.86 & 2.33 & 1.93 & 1.87\tabularnewline
1.89 & 1.35 & 1.88 & 1.53 & 1.40\tabularnewline
2.29 & 1.78 & 2.28 & 1.83 & 1.77\tabularnewline
2.02 & 1.37 & 2.01 & 1.48 & 1.35\tabularnewline
2.17 & 1.37 & 2.16 & 1.70 & 1.39\tabularnewline
\noalign{\smallskip}\hline\noalign{\smallskip}
\end{tabular}
\end{table}

When $\lambda$ is small, the result is nearly identical to the sample mean.  As $\lambda$ increases, the result moves towards the median value.

Unlike the sample average, however, this $M$-estimator will always converge to a value, even when the mean does not exist.  Consider the following draws from the absolute value of a Cauchy distribution.  The mean of these random variables is infinite so the sample average will not converge.  The median of this distribution is 1.

\begin{table}
\caption{Behavior of the sample mean, sample median, and $M$-estimator for 100 draws from the absolute value of a Cauchy distribution.  Repeated five times to show variation.}\label{TBL:2}
%
%
\begin{tabular}{p{1.6cm}p{1.6cm}p{1.6cm}p{1.6cm}p{1.6cm}}
\hline\noalign{\smallskip}
Mean & Median & $\lambda = 0.1$ & $\lambda = 1$ & $\lambda = 5$ 
 \\
\noalign{\smallskip}\svhline\noalign{\smallskip}
2.30 & 0.70 & 1.95 & 0.88 & 0.69\tabularnewline
2.70 & 1.10 & 2.56 & 1.27 & 1.11\tabularnewline
10.29 & 1.01 & 2.96 & 1.24 & 1.02\tabularnewline
3.59 & 1.09 & 2.58 & 1.32 & 1.13\tabularnewline
8.07 & 1.25 & 4.83 & 1.68 & 1.34\tabularnewline
\noalign{\smallskip}\hline\noalign{\smallskip}
\end{tabular}
\end{table}

Even for values such as $\lambda = 1$, the result is fairly close to the median.  For any $\lambda > 0$, the $M$-estimator will not go to infinity as the sample average does, but instead to converge to a fixed value as the number of samples goes to infinity.

\subsection{Timings}

To test the time required to create the new estimates, the algorithm presented here together with the algorithm from~\cite{hubertoappearb} were implemented in R.  Table~\ref{TBL:3} shows the results of running both algorithms using an Euler-Maruyama simulation of a simple SDE as a test case.

\begin{table}
\caption{Behavior of the sample mean, sample median, and $M$-estimator for 100 draws from the absolute value of a Cauchy distribution.  Repeated five times to show variation.}\label{TBL:3}
%
%
\begin{tabular}{p{1.6cm}p{1.6cm}p{1.6cm}p{1.6cm}p{1.6cm}}
\hline\noalign{\smallskip}
epsilon & delta & CG & New & relative change\tabularnewline
\noalign{\smallskip}\svhline\noalign{\smallskip}
0.10 & 1e-06 & 551.94 & 461.48 & -0.1638946\tabularnewline
0.05 & 1e-06 & 2013.24 & 1822.40 & -0.0947925\tabularnewline
\noalign{\smallskip}\hline\noalign{\smallskip}
\end{tabular}
\end{table}

As can be seen from the data, the relative change is near $-2\epsilon$, and the relative change gets closer to $-2\epsilon$ the smaller $\epsilon$ becomes.

\section{Conclusion}

The modified Catoni $M$-estimator presented here gives a means of interpolating between the sample mean and the sample average.  The estimator is designed for the output of Monte Carlo simulations where the distribution is usually unknown, but often it is possible to compute a bound on the relative standard deviation.  It is fast to calculate in practice and can be computed exactly in terms of square and cube roots in polynomial time.  The estimator has a parameter $\lambda$ which controls how close the estimate is to the sample mean or sample median.

Given a known upper bound $c$ on the relative standard deviation of the output, $\lambda$ can be chosen as $\epsilon(1-\epsilon^2)/(c^2+\epsilon^2)$ to yield an $(\epsilon,\delta)$-randomized approximation scheme that uses a number of samples (to first order) equal to that if the data was normally distributed.  Even if $c$ is unknown (or infinite), the estimator will still converge to a fixed measure of centrality for any choice of $\lambda$.

\appendix

\section{Code}

The following code was used to create the tables in Section 4.

The function \texttt{psi\_cg} implements the method from Catoni \&
Guilini {[}2{]}:

\begin{Shaded}
\begin{Highlighting}[]
\NormalTok{psi_cg <-}\StringTok{ }\NormalTok{function(m, x, lambda) \{}
  \CommentTok{# Calculates the psi function for a piece of data}
  \NormalTok{u <-}\StringTok{ }\NormalTok{lambda *}\StringTok{ }\NormalTok{(x /}\StringTok{ }\NormalTok{m -}\StringTok{ }\DecValTok{1}\NormalTok{)}
  \NormalTok{a1 <-}\StringTok{ }\DecValTok{2} \NormalTok{*}\StringTok{ }\KeywordTok{sqrt}\NormalTok{(}\DecValTok{2}\NormalTok{) /}\StringTok{ }\DecValTok{3}
  \NormalTok{a2 <-}\StringTok{ }\KeywordTok{sqrt}\NormalTok{(}\DecValTok{2}\NormalTok{)}
  \NormalTok{d <-}\StringTok{ }\NormalTok{a1 *}\StringTok{ }\NormalTok{(u >}\StringTok{ }\NormalTok{a2) -}\StringTok{ }\NormalTok{a2 *}\StringTok{ }\NormalTok{(u <}\StringTok{ }\NormalTok{-a2) +}\StringTok{ }\NormalTok{(u -}\StringTok{ }\NormalTok{u^}\DecValTok{3} \NormalTok{/}\StringTok{ }\DecValTok{6}\NormalTok{) *}\StringTok{ }\NormalTok{(u >=}\StringTok{ }\NormalTok{-}\DecValTok{1}\NormalTok{) *}\StringTok{ }\NormalTok{(u <=}\StringTok{ }\DecValTok{1}\NormalTok{)}
  \KeywordTok{return}\NormalTok{(}\KeywordTok{sum}\NormalTok{(d))}
\NormalTok{\}}
\end{Highlighting}
\end{Shaded}

\begin{Shaded}
\begin{Highlighting}[]
\NormalTok{psi_h <-}\StringTok{ }\NormalTok{function(m, x, lambda) \{}
  \CommentTok{# Calculates the psi function for a piece of data}
  \NormalTok{u <-}\StringTok{ }\NormalTok{lambda *}\StringTok{ }\NormalTok{(x /}\StringTok{ }\NormalTok{m -}\StringTok{ }\DecValTok{1}\NormalTok{)}
  \NormalTok{d <-}\StringTok{ }\NormalTok{(}\DecValTok{5} \NormalTok{/}\StringTok{ }\DecValTok{6}\NormalTok{) *}\StringTok{ }\NormalTok{(u >}\StringTok{ }\DecValTok{1}\NormalTok{) -}\StringTok{ }\NormalTok{(}\DecValTok{5} \NormalTok{/}\StringTok{ }\DecValTok{6}\NormalTok{) *}\StringTok{ }\NormalTok{(u <}\StringTok{ }\NormalTok{-}\DecValTok{1}\NormalTok{) +}\StringTok{ }\NormalTok{(u -}\StringTok{ }\NormalTok{u^}\DecValTok{3} \NormalTok{/}\StringTok{ }\DecValTok{6}\NormalTok{) *}\StringTok{ }\NormalTok{(u >=}\StringTok{ }\NormalTok{-}\DecValTok{1}\NormalTok{) *}\StringTok{ }\NormalTok{(u <=}\StringTok{ }\DecValTok{1}\NormalTok{)}
  \KeywordTok{return}\NormalTok{(}\KeywordTok{sum}\NormalTok{(d))}
\NormalTok{\}}
\end{Highlighting}
\end{Shaded}

\begin{Shaded}
\begin{Highlighting}[]
\NormalTok{m <-}\StringTok{ }\NormalTok{function(x, lambda, }\DataTypeTok{estimate =} \NormalTok{psi_h, }\DataTypeTok{tol =} \DecValTok{10}\NormalTok{^(-}\DecValTok{12}\NormalTok{)) \{}
  \NormalTok{x <-}\StringTok{ }\KeywordTok{sort}\NormalTok{(x) }\CommentTok{# Put everything in order}
  \NormalTok{a <-}\StringTok{ }\NormalTok{x[}\DecValTok{1}\NormalTok{]}
  \NormalTok{b <-}\StringTok{ }\NormalTok{x[}\KeywordTok{length}\NormalTok{(x)]}
  \NormalTok{while ((b -}\StringTok{ }\NormalTok{a) >}\StringTok{ }\NormalTok{tol) \{}
    \NormalTok{c <-}\StringTok{ }\NormalTok{(a +}\StringTok{ }\NormalTok{b) /}\StringTok{ }\DecValTok{2}
    \NormalTok{psi.c <-}\StringTok{ }\KeywordTok{estimate}\NormalTok{(c, x, lambda)}
    \NormalTok{if (psi.c <}\StringTok{ }\DecValTok{0}\NormalTok{)}
      \NormalTok{b <-}\StringTok{ }\NormalTok{c}
    \NormalTok{else}
      \NormalTok{a <-}\StringTok{ }\NormalTok{c}
  \NormalTok{\}}
  \KeywordTok{return}\NormalTok{((a +}\StringTok{ }\NormalTok{b) /}\StringTok{ }\DecValTok{2}\NormalTok{)}
\NormalTok{\}}
\end{Highlighting}
\end{Shaded}

\subsection{Numerical Experiments}\label{numerical-experiments}

Now we are ready to undertake the numerical experiments.

\begin{Shaded}
\begin{Highlighting}[]
\CommentTok{# install.packages("tidyverse")}
\KeywordTok{library}\NormalTok{(tidyverse)}
\end{Highlighting}
\end{Shaded}

Create Table 1 (accuracy new method):

\begin{Shaded}
\begin{Highlighting}[]
\CommentTok{# Generate data to work with}
\KeywordTok{set.seed}\NormalTok{(}\DecValTok{123456}\NormalTok{)}
\NormalTok{v <-}\StringTok{ }\KeywordTok{rexp}\NormalTok{(}\DecValTok{100} \NormalTok{*}\StringTok{ }\DecValTok{5}\NormalTok{, }\DataTypeTok{rate =} \DecValTok{1} \NormalTok{/}\StringTok{ }\DecValTok{2}\NormalTok{)}
\NormalTok{r <-}\StringTok{ }\KeywordTok{c}\NormalTok{(}\KeywordTok{rep}\NormalTok{(}\DecValTok{1}\NormalTok{, }\DecValTok{100}\NormalTok{), }\KeywordTok{rep}\NormalTok{(}\DecValTok{2}\NormalTok{, }\DecValTok{100}\NormalTok{), }\KeywordTok{rep}\NormalTok{(}\DecValTok{3}\NormalTok{, }\DecValTok{100}\NormalTok{), }\KeywordTok{rep}\NormalTok{(}\DecValTok{4}\NormalTok{, }\DecValTok{100}\NormalTok{), }\KeywordTok{rep}\NormalTok{(}\DecValTok{5}\NormalTok{, }\DecValTok{100}\NormalTok{))}
\NormalTok{df <-}\StringTok{ }\KeywordTok{tibble}\NormalTok{(}\DataTypeTok{run =} \NormalTok{r, }\DataTypeTok{rv =} \NormalTok{v)}
\NormalTok{df_sum <-}\StringTok{ }\NormalTok{df %>%}
\StringTok{  }\KeywordTok{group_by}\NormalTok{(run) %>%}
\StringTok{  }\KeywordTok{summarize}\NormalTok{(}
    \DataTypeTok{Mean =} \KeywordTok{mean}\NormalTok{(rv), }
    \DataTypeTok{Median =} \KeywordTok{median}\NormalTok{(rv),}
    \StringTok{"$}\CharTok{\textbackslash{}\textbackslash{}}\StringTok{lambda$ = 0.1"} \NormalTok{=}\StringTok{ }\KeywordTok{m}\NormalTok{(rv, }\DataTypeTok{lambda =} \FloatTok{0.1}\NormalTok{),}
    \StringTok{"$}\CharTok{\textbackslash{}\textbackslash{}}\StringTok{lambda$ = 1"} \NormalTok{=}\StringTok{ }\KeywordTok{m}\NormalTok{(rv, }\DataTypeTok{lambda =} \DecValTok{1}\NormalTok{),}
    \StringTok{"$}\CharTok{\textbackslash{}\textbackslash{}}\StringTok{lambda$ = 5"} \NormalTok{=}\StringTok{ }\KeywordTok{m}\NormalTok{(rv, }\DataTypeTok{lambda =} \DecValTok{5}\NormalTok{),}
  \NormalTok{) %>%}
\StringTok{  }\KeywordTok{select}\NormalTok{(-run)}
\KeywordTok{kable}\NormalTok{(df_sum, }\DataTypeTok{digits =} \DecValTok{2}\NormalTok{)}
\end{Highlighting}
\end{Shaded}

\begin{longtable}[]{@{}rrrrr@{}}
\toprule
Mean & Median & \(\lambda\) = 0.1 & \(\lambda\) = 1 & \(\lambda\) =
5\tabularnewline
\midrule
\endhead
2.34 & 1.86 & 2.33 & 1.93 & 1.87\tabularnewline
1.89 & 1.35 & 1.88 & 1.53 & 1.40\tabularnewline
2.29 & 1.78 & 2.28 & 1.83 & 1.77\tabularnewline
2.02 & 1.37 & 2.01 & 1.48 & 1.35\tabularnewline
2.17 & 1.37 & 2.16 & 1.70 & 1.39\tabularnewline
\bottomrule
\end{longtable}

Create Table 2 (accuracy new method):

\begin{Shaded}
\begin{Highlighting}[]
\CommentTok{# Generate data to work with}
\KeywordTok{set.seed}\NormalTok{(}\DecValTok{123456}\NormalTok{)}
\NormalTok{v <-}\StringTok{ }\KeywordTok{abs}\NormalTok{(}\KeywordTok{rcauchy}\NormalTok{(}\DecValTok{100} \NormalTok{*}\StringTok{ }\DecValTok{5}\NormalTok{))}
\NormalTok{r <-}\StringTok{ }\KeywordTok{c}\NormalTok{(}\KeywordTok{rep}\NormalTok{(}\DecValTok{1}\NormalTok{, }\DecValTok{100}\NormalTok{), }\KeywordTok{rep}\NormalTok{(}\DecValTok{2}\NormalTok{, }\DecValTok{100}\NormalTok{), }\KeywordTok{rep}\NormalTok{(}\DecValTok{3}\NormalTok{, }\DecValTok{100}\NormalTok{), }\KeywordTok{rep}\NormalTok{(}\DecValTok{4}\NormalTok{, }\DecValTok{100}\NormalTok{), }\KeywordTok{rep}\NormalTok{(}\DecValTok{5}\NormalTok{, }\DecValTok{100}\NormalTok{))}
\NormalTok{df <-}\StringTok{ }\KeywordTok{tibble}\NormalTok{(}\DataTypeTok{run =} \NormalTok{r, }\DataTypeTok{rv =} \NormalTok{v)}
\NormalTok{df_sum <-}\StringTok{ }\NormalTok{df %>%}
\StringTok{  }\KeywordTok{group_by}\NormalTok{(run) %>%}
\StringTok{  }\KeywordTok{summarize}\NormalTok{(}
    \DataTypeTok{Mean =} \KeywordTok{mean}\NormalTok{(rv), }
    \DataTypeTok{Median =} \KeywordTok{median}\NormalTok{(rv),}
    \StringTok{"$}\CharTok{\textbackslash{}\textbackslash{}}\StringTok{lambda$ = 0.1"} \NormalTok{=}\StringTok{ }\KeywordTok{m}\NormalTok{(rv, }\DataTypeTok{lambda =} \FloatTok{0.1}\NormalTok{),}
    \StringTok{"$}\CharTok{\textbackslash{}\textbackslash{}}\StringTok{lambda$ = 1"} \NormalTok{=}\StringTok{ }\KeywordTok{m}\NormalTok{(rv, }\DataTypeTok{lambda =} \DecValTok{1}\NormalTok{),}
    \StringTok{"$}\CharTok{\textbackslash{}\textbackslash{}}\StringTok{lambda$ = 5"} \NormalTok{=}\StringTok{ }\KeywordTok{m}\NormalTok{(rv, }\DataTypeTok{lambda =} \DecValTok{5}\NormalTok{),}
  \NormalTok{) %>%}
\StringTok{  }\KeywordTok{select}\NormalTok{(-run)}
\KeywordTok{kable}\NormalTok{(df_sum, }\DataTypeTok{digits =} \DecValTok{2}\NormalTok{)}
\end{Highlighting}
\end{Shaded}

\begin{longtable}[]{@{}rrrrr@{}}
\toprule
Mean & Median & \(\lambda\) = 0.1 & \(\lambda\) = 1 & \(\lambda\) =
5\tabularnewline
\midrule
\endhead
2.30 & 0.70 & 1.95 & 0.88 & 0.69\tabularnewline
2.70 & 1.10 & 2.56 & 1.27 & 1.11\tabularnewline
10.29 & 1.01 & 2.96 & 1.24 & 1.02\tabularnewline
3.59 & 1.09 & 2.58 & 1.32 & 1.13\tabularnewline
8.07 & 1.25 & 4.83 & 1.68 & 1.34\tabularnewline
\bottomrule
\end{longtable}

Toy example of process where we know ratio of the variance to square of
mean, Brownian motion created through an inefficient process.

\begin{Shaded}
\begin{Highlighting}[]
\CommentTok{# Euler-Maruyama for SDE:  dS_t = dt + dW_t}
\NormalTok{bm <-}\StringTok{ }\NormalTok{function(t, h) \{}
  \NormalTok{z <-}\StringTok{ }\KeywordTok{rnorm}\NormalTok{(t /}\StringTok{ }\NormalTok{h)}
  \NormalTok{s <-}\StringTok{ }\DecValTok{0}
  \NormalTok{for (i in }\DecValTok{1}\NormalTok{:}\KeywordTok{length}\NormalTok{(z))}
    \NormalTok{s <-}\StringTok{ }\NormalTok{s +}\StringTok{ }\NormalTok{h *}\StringTok{ }\DecValTok{1} \NormalTok{+}\StringTok{ }\KeywordTok{sqrt}\NormalTok{(h) *}\StringTok{ }\NormalTok{z[i]}
  \KeywordTok{return}\NormalTok{(s)}
\NormalTok{\}}
\end{Highlighting}
\end{Shaded}

Time test for old:

\begin{Shaded}
\begin{Highlighting}[]
\NormalTok{time_psi_cg <-}\StringTok{ }\NormalTok{function(epsilon, delta, c) \{}
  \NormalTok{start_time <-}\StringTok{ }\KeywordTok{as.numeric}\NormalTok{(}\KeywordTok{as.numeric}\NormalTok{(}\KeywordTok{Sys.time}\NormalTok{())*}\DecValTok{1000}\NormalTok{, }\DataTypeTok{digits=}\DecValTok{15}\NormalTok{)}
  \NormalTok{n <-}\StringTok{ }\KeywordTok{ceiling}\NormalTok{(}\DecValTok{2} \NormalTok{*}\StringTok{ }\NormalTok{c^}\DecValTok{2} \NormalTok{*}\StringTok{ }\NormalTok{epsilon^(-}\DecValTok{2}\NormalTok{) *}\StringTok{ }\KeywordTok{log}\NormalTok{(}\DecValTok{2} \NormalTok{/}\StringTok{ }\NormalTok{delta) *}\StringTok{ }\NormalTok{(}\DecValTok{1} \NormalTok{+}\StringTok{ }\NormalTok{epsilon)^}\DecValTok{2}\NormalTok{)}
  \NormalTok{lambda <-}\StringTok{ }\NormalTok{epsilon /}\StringTok{ }\NormalTok{(}\DecValTok{1} \NormalTok{+}\StringTok{ }\NormalTok{epsilon) /}\StringTok{ }\NormalTok{c^}\DecValTok{2}
  \NormalTok{r <-}\StringTok{ }\KeywordTok{replicate}\NormalTok{(n, }\KeywordTok{bm}\NormalTok{(}\DecValTok{1}\NormalTok{, }\FloatTok{0.001}\NormalTok{))}
  \NormalTok{hatm <-}\StringTok{ }\KeywordTok{m}\NormalTok{(r, lambda, psi_cg) }
  \NormalTok{end_time <-}\StringTok{ }\KeywordTok{as.numeric}\NormalTok{(}\KeywordTok{as.numeric}\NormalTok{(}\KeywordTok{Sys.time}\NormalTok{())*}\DecValTok{1000}\NormalTok{, }\DataTypeTok{digits=}\DecValTok{15}\NormalTok{)}
  \KeywordTok{return}\NormalTok{(end_time -}\StringTok{ }\NormalTok{start_time)}
\NormalTok{\}}
\end{Highlighting}
\end{Shaded}

Time test for new:

\begin{Shaded}
\begin{Highlighting}[]
\NormalTok{time_psi_h <-}\StringTok{ }\NormalTok{function(epsilon, delta, c) \{}
  \NormalTok{start_time <-}\StringTok{ }\KeywordTok{as.numeric}\NormalTok{(}\KeywordTok{as.numeric}\NormalTok{(}\KeywordTok{Sys.time}\NormalTok{())*}\DecValTok{1000}\NormalTok{, }\DataTypeTok{digits=}\DecValTok{15}\NormalTok{)}
  \NormalTok{n <-}\StringTok{ }\KeywordTok{ceiling}\NormalTok{(}\DecValTok{2} \NormalTok{*}\StringTok{ }\NormalTok{(c^}\DecValTok{2} \NormalTok{*}\StringTok{ }\NormalTok{epsilon^(-}\DecValTok{2}\NormalTok{) +}\StringTok{ }\DecValTok{1}\NormalTok{) *}\StringTok{ }\NormalTok{(}\DecValTok{1} \NormalTok{+}\StringTok{ }\NormalTok{epsilon^}\DecValTok{2}\NormalTok{)^(-}\DecValTok{1}\NormalTok{) *}\StringTok{ }\KeywordTok{log}\NormalTok{(}\DecValTok{2} \NormalTok{/}\StringTok{ }\NormalTok{delta))}
  \NormalTok{lambda <-}\StringTok{ }\NormalTok{epsilon *}\StringTok{ }\NormalTok{(}\DecValTok{1} \NormalTok{-}\StringTok{ }\NormalTok{epsilon^}\DecValTok{2}\NormalTok{) /}\StringTok{ }\NormalTok{(c^}\DecValTok{2} \NormalTok{+}\StringTok{ }\NormalTok{epsilon^}\DecValTok{2}\NormalTok{)}
  \NormalTok{r <-}\StringTok{ }\KeywordTok{replicate}\NormalTok{(n, }\KeywordTok{bm}\NormalTok{(}\DecValTok{1}\NormalTok{, }\FloatTok{0.001}\NormalTok{))}
  \NormalTok{hatm <-}\StringTok{ }\KeywordTok{m}\NormalTok{(r, lambda, psi_h) }
  \NormalTok{end_time <-}\StringTok{ }\KeywordTok{as.numeric}\NormalTok{(}\KeywordTok{as.numeric}\NormalTok{(}\KeywordTok{Sys.time}\NormalTok{())*}\DecValTok{1000}\NormalTok{, }\DataTypeTok{digits=}\DecValTok{15}\NormalTok{)}
  \KeywordTok{return}\NormalTok{(end_time -}\StringTok{ }\NormalTok{start_time)}
\NormalTok{\}}
\end{Highlighting}
\end{Shaded}

Timings using an Intel(R) Core(TM) i7-6700 CPT @ 3.40GHz.

\begin{Shaded}
\begin{Highlighting}[]
\KeywordTok{set.seed}\NormalTok{(}\DecValTok{123456}\NormalTok{)}
\NormalTok{ep <-}\StringTok{ }\KeywordTok{c}\NormalTok{(}\FloatTok{0.1}\NormalTok{, }\FloatTok{0.05}\NormalTok{)}
\NormalTok{del <-}\StringTok{ }\KeywordTok{c}\NormalTok{(}\DecValTok{10}\NormalTok{^(-}\DecValTok{6}\NormalTok{), }\DecValTok{10}\NormalTok{^(-}\DecValTok{6}\NormalTok{))}
\NormalTok{r <-}\StringTok{ }\KeywordTok{rep}\NormalTok{(}\DecValTok{0}\NormalTok{, }\DecValTok{2}\NormalTok{)}
\NormalTok{s <-}\StringTok{ }\KeywordTok{rep}\NormalTok{(}\DecValTok{0}\NormalTok{, }\DecValTok{2}\NormalTok{)}
\NormalTok{for (i in }\DecValTok{1}\NormalTok{:}\DecValTok{2}\NormalTok{) \{}
  \NormalTok{r[i] <-}\StringTok{ }\KeywordTok{mean}\NormalTok{(}\KeywordTok{replicate}\NormalTok{(}\DecValTok{100}\NormalTok{, }\KeywordTok{time_psi_cg}\NormalTok{(ep[i], del[i], }\DecValTok{1}\NormalTok{)))}
  \NormalTok{s[i] <-}\StringTok{ }\KeywordTok{mean}\NormalTok{(}\KeywordTok{replicate}\NormalTok{(}\DecValTok{100}\NormalTok{, }\KeywordTok{time_psi_h}\NormalTok{(ep[i], del[i], }\DecValTok{1}\NormalTok{)))}
\NormalTok{\}}
\NormalTok{df <-}\StringTok{ }\KeywordTok{tibble}\NormalTok{(}\DataTypeTok{epsilon =} \NormalTok{ep, }\DataTypeTok{delta =} \NormalTok{del, }\DataTypeTok{CG =} \NormalTok{r, }\DataTypeTok{New =} \NormalTok{s) %>%}\StringTok{ }
\StringTok{  }\KeywordTok{mutate}\NormalTok{(}\DataTypeTok{improvement =} \NormalTok{(s -}\StringTok{ }\NormalTok{r) /}\StringTok{ }\NormalTok{r)}
\KeywordTok{kable}\NormalTok{(df)}
\end{Highlighting}
\end{Shaded}

\begin{longtable}[]{@{}rrrrr@{}}
\toprule
epsilon & delta & CG & New & improvement\tabularnewline
\midrule
\endhead
0.10 & 1e-06 & 551.94 & 461.48 & -0.1638946\tabularnewline
0.05 & 1e-06 & 2013.24 & 1822.40 & -0.0947925\tabularnewline
\bottomrule
\end{longtable}

\begin{acknowledgement}
This work supported by National Science Foundation grant DMS-1418495.
\end{acknowledgement}

%
%

\end{document}